\documentclass[journal]{IEEEtran}

\usepackage{amsmath,amssymb,amsthm}
\usepackage{graphicx}
\usepackage{cite}
\usepackage{booktabs}
\usepackage{xcolor}
\usepackage{mathtools}
\usepackage{bbm}
\usepackage[utf8]{inputenc}
\usepackage[T1]{fontenc}
\usepackage{comment}
\usepackage{tabularx}
\usepackage{dsfont}
\usepackage{bbold}

\DeclareMathOperator{\diag}{diag}
\DeclareMathOperator{\sign}{sign}
\newcommand{\Real}[1]{\operatorname{Re}\!\left(#1\right)}
\newcommand{\Imag}[1]{\operatorname{Im}\!\left(#1\right)}
\newcommand {\red}[1]{{\color{red} #1}}

\newtheorem{theorem}{Theorem}

\newtheorem{property}{\red{Property}}
\newtheorem{remark}{Remark}

\renewcommand{\red}[1]{#1}
\begin{document}

\title{\red{Switched Control of Complex-Valued \\ Kuramoto Networks}}

\author{Lorenzo~Giordano,~Josep~M.~Olm,~and~Mario~di~Bernardo%
\thanks{L.~Giordano and M.~di Bernardo are with the Scuola Superiore Meridionale, 80138 Naples, Italy (e-mail: l.giordano@ssmeridionale.it).}%
\thanks{J.\,M.~Olm is with the Department of Mathematics \& Institute of Industrial and Control Engineering, Universitat Polit\`{e}cnica de Catalunya, 08028 Barcelona, Spain (e-mail: josep.olm@upc.edu). Corresponding author.}%
\thanks{M.~di~Bernardo is also with the Department of Electrical Engineering and Information Technologies, Universit\`{a} degli Studi di Napoli Federico~II, 80125 Naples, Italy (e-mail: mario.dibernardo@unina.it).}%
\thanks{The work of J.\,M.~Olm was partially supported by the Government of Spain through Project PID2021-122821NB-I00 funded by MICIU/AEI/10.13039/501100011033 and ERDF/EU, and by the Generalitat de Catalunya through Project~2021~SGR~00376.}}

\maketitle
\thispagestyle{empty}

% ================================================================

\begin{abstract}
Synchronization in networks of coupled oscillators is classically studied via the Kuramoto model, whose intrinsic nonlinearity limits analytical tractability and complicates control design. Complex-valued extensions circumvent this by embedding phase dynamics into a higher-dimensional linear state space, where regulating complex-state moduli to a common value \red{reproduces the nominal Kuramoto phase dynamics on an invariant manifold}.
Existing approaches to address this problem correspond, within a control framework, to state-feedback and hybrid reset-based strategies, each with performance constraints. We propose two switched control designs that overcome these limitations: a switched feedforward law ensuring exact phase correspondence at all times, and a feedforward plus sliding-mode law achieving finite-time \red{mapping onto the original Kuramoto dynamics} without spectral gain tuning. Additionally, we present a non-autonomous complex-valued MIMO sliding-mode controller that enforces phase locking at a prescribed frequency in finite time, \red{regardless} of natural frequencies and coupling strengths, \red{under appropriate control gain selection}. Simulations confirm improved transient response, steady-state accuracy, and robustness, including synchronization of heterogeneous networks where the classical real-valued Kuramoto model fails.
\end{abstract}


% ================================================================
\section{Introduction}
\label{sec:intro}
Synchronization in networks of coupled oscillators arises across many disciplines, from physics and engineering to biology and neuroscience \cite{dorfler_2014}. A seminal formulation of this collective behavior was introduced by Kuramoto in 1975 \cite{kuramoto_1975}, and the Kuramoto model has since become a cornerstone for studying phase synchronization. It has been widely applied to real-world systems, including flashing fireflies \cite{mirollo_1990}, Josephson junction arrays \cite{wiesenfeld_1996}, chemical oscillators \cite{kiss_2002}, neural dynamics \cite{cabral_2011}, and power networks \cite{Guo_2021}.

Despite its broad use, the Kuramoto model remains analytically and computationally challenging due to its intrinsic nonlinearity, with closed-form solutions only in special cases \cite{muller_2021}. To address this, complex-valued extensions embed the nonlinear phase dynamics into a higher-dimensional \emph{linear} complex state space \cite{muller_2021,budzinski_2022}. These map oscillator phases to complex states that evolve linearly in time, facilitating spectral and modal analysis, while their arguments reproduce the Kuramoto dynamics.

This linear reformulation is crucial from a control perspective, enabling linear control tools to shape collective phase behavior and allowing systematic controller design for synchronization tasks otherwise intractable in the nonlinear setting. Beyond control, such embeddings have been applied to equilibrium classification \cite{nguyen_2023}, complex-valued neural networks \cite{budzinski_2024_exact}, traveling waves in neuroscience \cite{benigno_2023_waves}, image segmentation \cite{liboni_2025_image}, and quantum information theory \cite{amati_2025_quantum}. \red{Similarly, related complexified states and complex weights have been utilized to investigate fundamental synchrony bounds \cite{thumler_2023, bottcher_2024}}.

Two main control-oriented linear reformulations have emerged in the literature. Roberts \cite{roberts_2008} augments the linear complex model with an input that drives magnitudes to a common value; as they converge, the argument dynamics asymptotically map onto the nominal Kuramoto phases, allowing order parameter computation and stability analysis \cite{panteley_2013_2,panteley_2013_1}. Budzinski et al.~\cite{budzinski_2022} instead apply periodic magnitude resets while exploiting the linear solution between them, yielding close agreement with real-valued trajectories over finite time horizons.

Although developed independently, both methods share the same \emph{control objective}: regulate magnitudes so that phase arguments reproduce classical Kuramoto behavior. From a control viewpoint, \cite{roberts_2008} is a continuous state-feedback strategy, while \cite{budzinski_2022} is a hybrid reset-based law. Both exhibit limitations: the former ensures only asymptotic \red{mapping} and requires spectral gain tuning, \red{i.e. eigenvalue shifting}, while the latter achieves exact matching only at reset instants, with magnitude drift accumulating between them.

Motivated by this unified perspective, this work leverages the linear complex-valued formulation to address two control problems for the Kuramoto model. First, we aim to reproduce the phase dynamics of the original real-valued system, improving upon \cite{roberts_2008} and \cite{budzinski_2022}. Second, we seek to enforce phase locking at a prescribed frequency, regardless of natural frequencies and coupling strengths, achieving robust and controlled synchronization via complex-valued sliding-mode control (SMC) \cite{doria1}. Hence, the main contributions of this work are:

\begin{enumerate}
  \item A \emph{switched feedforward} law that achieves exact, continuous phase matching with the real-valued Kuramoto dynamics for all $t\ge0$ under unit initial magnitudes, overcoming the discrete-instant limitation of \cite{budzinski_2022}.
  \item A \emph{feedforward\,+\,SMC} law that upgrades the asymptotic
    \red{mapping} of \cite{roberts_2008} to finite-time
    \red{mapping}, with an explicit and adjustable reaching-time
    bound. No spectral gain tuning is required.  
  \item An extension of the complex-valued MIMO SMC framework
    of \cite{diPalo_2023} to \emph{non-autonomous} systems \red{that, under appropriate control gain values,
    enforces} phase locking at a prescribed frequency
    in finite time, \red{irrespectively} of natural frequencies and
    coupling strengths---including heterogeneous networks where
    the classical real-valued model fails to synchronize.
\end{enumerate}

The remainder of the article is organized as follows. Section~\ref{sec:notation} introduces notation. Section~\ref{sec:model} reviews the Kuramoto model and its complex extension. Section~\ref{sec:sota} revisits existing approaches from a control perspective. Section~\ref{sec:ff} presents the proposed feedforward and feedforward plus SMC methods. Section~\ref{sec:csmc} develops the prescribed-frequency SMC strategy. Section~\ref{sec:num} provides numerical validation, and Section~\ref{sec:conc} concludes the paper.

% ================================================================
\section{Notation}
\label{sec:notation}

For $z = z^R + i z^I \in \mathbb{C}$, $z^R=\Real{z}$ and $z^I=\Imag{z}$ are the real and imaginary parts of $z$, respectively. In turn, $z^*$ is the conjugate of $z$, $|z|=\sqrt{z^* z}$ its magnitude, and $\phi_z$ its phase.  

For $w \in \mathbb{C}^n$, $\Real{w}$ and $\Imag{w}$ are vectors that collect real and imaginary parts of $w$, respectively, $w^*$ is the conjugate transpose of $w$, and $\lvert w \rvert_\bullet$, $\phi_w$ are vectors of magnitudes and phases, respectively.
Moreover, $\cos(\phi_w) \in \mathbb{R}^n$ and $e^{i \phi_w}\in\mathbb{C}^n$ are componentwise defined functions. 

The complex signum function maps each component of a vector to a unit-magnitude value while preserving its original phase \cite{doria1}, i.e. 
$\sign(w) \coloneqq \left(\sign\left(w_k\right)\right) \in \mathbb{C}^n$ with
\begin{equation*} 
    \sign\left(w_k\right)=\left\{\begin{array}{cl} \frac{w_k}{\left|w_k\right|} & w_k \neq 0, \\ 0 & w_k = 0.  \\  \end{array}\right.
\end{equation*}

The 1-norm is defined as $\|w\|_1=\sum_{k=1}^n |w_k|$, the 2-norm as $\|w\|_2=\sqrt{w^\ast w}$, and the $\infty$-norm as $\left\|w\right\|_\infty=\sup_k\{\left|w_k\right|\}$. Finally, $\mathds{1}_n$ denotes the $n$-dimensional vector of ones.

% ================================================================
\section{The Kuramoto Model and Complex-Valued Extension}
\label{sec:model}

Given a undirected graph $\mathcal{G}=(\mathcal{V},
\mathcal{E})$ with $N$ nodes, the Kuramoto model \cite{kuramoto_1975} reads as
\begin{equation}
  \dot{\theta}_k = \omega_k +
    \sigma\!\sum_{j=1}^N a_{kj}\sin(\theta_j-\theta_k),
  \label{eq:real_kuramoto}
\end{equation}
where $\theta_k\in\mathbb{R}$ are oscillator phases,
$\omega_k\in\mathbb{R}$ are natural frequencies \red{with $\omega_k \neq \omega_j$ and $\dot{\omega}_k=0$, $\forall k,j \in \left\{1,\ldots,N\right\}$,}  $\sigma>0$ is the coupling strength, \red{and $a_{kj}$ are the elements of the adjacency matrix $A$. Specifically, $a_{kj} = 1$ if there is an edge from $k$ to $j$, and $a_{kj} = 0$ otherwise}. 

Synchronization is measured by the order parameter 
$r(t)=\frac{1}{N}\sum_{k}e^{i\theta_k(t)}$, where $|r|=1$ indicates perfect synchronization.

Subtracting $i \cos(\theta_j{-}\theta_k)$ from the
coupling in \eqref{eq:real_kuramoto}, performing straightforward algebraic manipulation
and setting $x=e^{i\theta}$, yields the \red{complex-valued} linear system
\cite{muller_2021}:
\begin{equation} \label{eq:complex_lti}
  \dot{x} = \left(i\diag(\omega)+\sigma A\right)x.
\end{equation}
Splitting \eqref{eq:complex_lti} into magnitude and argument
dynamics gives 
\begin{subequations} \label{eq:modarg_ku}
\begin{align}
\label{eq:mo_ku} \frac{\mathrm{d}}{\mathrm{d}t}\lvert x_k \rvert  &=f_k(x) \coloneqq \sigma \sum_{j=1}^N a_{kj}\lvert x_j \rvert\cos\!\left(\phi_{x_j}-\phi_{x_k}\right), \\
\label{eq:arg_ku}\frac{\mathrm{d}}{\mathrm{d}t}\phi_{x_k} &=g_k(x) \coloneqq \omega_k+\sigma \sum_{j=1}^N a_{kj}\tfrac{\lvert x_j \rvert}{\lvert x_k \rvert}\sin\!\left(\phi_{x_j}-\phi_{x_k}\right).
\end{align}
\end{subequations}
for all $k$ \red{\cite{budzinski_2022}. Notice that letting $\theta_k=\theta_k^R+i\theta_k^I \in \mathbb{C}$, it is $x_k=e^{i\theta_k}=\lvert x_k \rvert e^{i\phi_{x_k}}$, with $\lvert x_k \rvert=e^{-\theta_k^I}$, $\phi_{x_k}=\theta_k^R$. Hence, we have $\lvert x_k(t) \rvert \neq 0$, $\forall t \geq 0$, by construction. Moreover,} if $|x_k|\to c$ for all $k$, then \eqref{eq:arg_ku} \red{maps onto} \eqref{eq:real_kuramoto}. This can be achieved by introducing a control input $u\in\mathbb{C}^N$ in \eqref{eq:complex_lti}:
\begin{equation} \label{eq:complex_kuramoto_vectorial_u}
  \dot{x} = \left(i\diag(\omega)+\sigma A\right)x + u.
\end{equation}

% ================================================================
\section{State-of-the-Art: A Control Perspective}
\label{sec:sota}

This section reviews the existing linear reformulations of
the Kuramoto model from a control-theoretic standpoint, making
explicit the shared objective that motivates the present work.

\subsection{State-feedback control}
\label{subsec:sf}

In~\cite{roberts_2008} the following linear 
system on $\mathbb{C}^N$ is introduced: 
\begin{equation}
  \dot{x}=\left(i\diag(\omega)-\diag(\mu)+\sigma A\right)x,
  \label{eq:roberts}
\end{equation}
where $\mu\in\mathbb{R}^N$ is chosen so that suitable spectral
conditions are satisfied, enabling the arguments of
\eqref{eq:roberts} to \red{asymptotically 
map onto} the dynamics of the real-valued
Kuramoto model \eqref{eq:real_kuramoto}. 
Specifically, the
matrix in \eqref{eq:roberts} must have a single purely imaginary
eigenvalue where all the entries of the associated eigenvector have equal modulus, while the remaining eigenvectors must have strictly negative
real part. Under these conditions, the argument dynamics of
\eqref{eq:roberts} asymptotically reproduce 
\eqref{eq:real_kuramoto} \red{on an invariant manifold}.

The system \eqref{eq:roberts} can be interpreted as the
closed-loop dynamics of \eqref{eq:complex_kuramoto_vectorial_u} under the
state-feedback law $u=-\diag(\mu)x$. The control design
problem then consists in selecting $\mu$ to satisfy the
spectral constraints above.

\subsection{Hybrid reset control}
\label{subsec:hybrid}

The method in \cite{budzinski_2022} uses a time-window approach exploiting the linearity of \eqref{eq:complex_lti}: time is divided into intervals of length $T$, where at each reset the state $x$ is projected onto the unit circle while preserving phases, and then evolves linearly. This enables analytic propagation  over each window via matrix exponentials, avoiding numerical integration and improving computational efficiency.

From a hybrid-control viewpoint, this is described by the flow--jump system
\begin{subequations} \label{eq:hy_flow}
\begin{align}
\dot x &= \left( i \,\mathrm{diag}(\omega) + \sigma A \right)x, \\
\dot \tau &= 1,
\end{align}
\end{subequations}
for $(x,\tau)\in C=\{(x,\tau)\in\mathbb{C}^N\times\mathbb{R}_{\ge0}:0\le\tau<T\}$, with jumps at $\tau=T$ given by
\begin{subequations} \label{eq:hy_jump}
\begin{align}
x^+ &= \mathrm{sign}(x(t)), \\
\tau^+ &= 0,
\end{align}
\end{subequations}
for $(x,\tau)\in D=\{(x,\tau)\in\mathbb{C}^N\times\mathbb{R}_{\ge0}:\tau=T\}$.

This construction matches the \red{nominal} Kuramoto dynamics over finite, nontrivial time  intervals, with accuracy improving for smaller $T$ due to reduced inter-reset amplitude drift.

\subsection{Unified perspective}
\label{subsec:unified}

Although the approaches in \cite{roberts_2008} and \cite{budzinski_2022} were developed independently, they share a common control objective: enforce a uniform modulus on the complex states of \eqref{eq:complex_kuramoto_vectorial_u} so that their arguments \red{map onto} the phase dynamics of \eqref{eq:real_kuramoto}. Both exhibit fundamental performance limitations: an asymptotic mapping with constrained gain selection in the former, and discrete-time matching with inter-reset magnitude drift in the latter. Table~\ref{tab:summary} compares these methods with the three controllers proposed here, highlighting the progressive improvements in \red{mapping performance} achieved by the switched and sliding-mode designs introduced in the following sections.

\begin{table}
\centering
\caption{Overview of controllers for the complex-valued
         Kuramoto network \eqref{eq:complex_kuramoto_vectorial_u}.
         FF\,=\,feedforward; SMC\,=\,sliding-mode control.}
\label{tab:summary}
\renewcommand{\arraystretch}{1.15}
\begin{tabularx}{\columnwidth}{@{}l>{\raggedright\arraybackslash}X>{\raggedright\arraybackslash}X@{}}
\toprule
\textbf{Controller} & \textbf{Objective} & \textbf{\red{Mapping}} \\
\midrule
State-feedback \cite{roberts_2008}
  & Phase replication & Asymptotic \\
Hybrid reset \cite{budzinski_2022}
  & Phase replication & Exact at reset instants \\
Switched FF (Thm.\,\ref{th:sfc})
  & Phase replication & Exact $\forall t\ge 0$ \\
FF\,+\,SMC (Thm.\,\ref{th:smc_ff})
  & Phase replication & Finite-time \\
Complex SMC (Thm.\,\ref{th:smc_particularized})
  & Synchronization at prescribed frequency & Finite-time\\
\bottomrule
\end{tabularx}
\end{table}

% ================================================================
\section{Switched Controllers}
\label{sec:ff}

\subsection{Switched feedforward control}
\label{subsec:ff}

First, we introduce a switched feedforward controller that yields an exact correspondence with \eqref{eq:real_kuramoto} for all time, provided that the magnitudes of the state of \eqref{eq:complex_kuramoto_vectorial_u} are initialized at unity. Recalling the vector fields $f_k,g_k:\mathbb{C}^N \to\mathbb{R}$ defined via the right-hand sides of \eqref{eq:modarg_ku}, the control law is presented next.

\begin{theorem}\label{th:sfc}
Let $u=(|u_k|e^{i\phi_{u_k}})\in\mathbb{C}^N$ with
\begin{subequations}\label{eq:u_s}
\begin{align}
  |u_k| &= |f_k(x)|, \label{eq:u_s_mag}\\
  \phi_{u_k} &= \phi_{x_k}+\tfrac{\pi}{2}
    \left(1+\sign(f_k(x))\right). \label{eq:u_s_arg}
\end{align}
\end{subequations}
If $|x_k(0)|=1$ for all $k$, then $|x_k(t)|\equiv1$ for all
$t\ge0$ and the argument dynamics of \eqref{eq:complex_kuramoto_vectorial_u}
coincide exactly with \eqref{eq:real_kuramoto}.
\end{theorem}

\begin{proof} A direct computation shows that system \eqref{eq:complex_kuramoto_vectorial_u} can be decomposed for all $k \in \{1,\ldots,N\}$ into
\begin{subequations} \label{modarg_smc1}
\begin{align}
\label{eq:mod_smc1}  \frac{\mathrm{d}}{\mathrm{d}t} \lvert x_k \rvert  =& \, f_k(x) 
 +\lvert u_k \rvert \cos\!\left(\phi_{u_k}-\phi_{x_k}\right), \\ 
 \label{eq:arg_smc1} \frac{\mathrm{d}}{\mathrm{d}t} \phi_{x_k}  =& \, g_k(x) + \frac{\lvert u_k \rvert}{\lvert x_k \rvert}\sin\!\left(\phi_{u_k}-\phi_{x_k}\right).
\end{align}
\end{subequations}
Under the control law \eqref{eq:u_s}, \red{applying Filippov regularization to \eqref{eq:mod_smc1} yields the radial differential inclusion $$\frac{\mathrm{d}}{\mathrm{d}t}\lvert x_k \rvert \in f_k(x) - \lvert f_k(x) \rvert \operatorname{Sgn} \left(f_k(x)\right) = \{0\},$$
where $\operatorname{Sgn}(\cdot)$ is the set-valued signum map. }
Hence, when initialized with $\lvert x(0) \rvert_{\bullet} = \mathds{1}_N$, then $\lvert x(t) \rvert_{\bullet} = \mathds{1}_N$ for all $t\geq 0$. Substituting this identity into \eqref{eq:arg_smc1} yields
\begin{equation} 
\left. \frac{\mathrm{d}}{\mathrm{d}t} \phi_{x_k} \right|_{\lvert x \rvert_{\bullet} = \mathds{1}_N} = \omega_k+\sigma \sum_{j=1}^N a_{kj}\sin\!\left(\phi_{x_j}-\phi_{x_k}\right),
\label{eq:controlled_arguments}
\end{equation}
because \red{$\lvert f_k(x) \rvert \sin\!\left(\frac{\pi}{2}\left(1+\operatorname{Sgn}\left(f_k(x)\right)\right)\right)=\{0\}$}, $\forall \, x \in \mathbb{C}^N$. Hence, the dynamics of the arguments match those of the real Kuramoto model \eqref{eq:real_kuramoto} exactly.
\end{proof}

\begin{remark}
\red{1) Theorem~\ref{th:sfc} improves upon \cite{budzinski_2022} by ensuring phase correspondence holds continuously for all $t \ge 0$, rather than only at discrete reset instants. The trade-off involves the requirement of continuous nonlinear feedback control and the numerical integration of the closed-loop dynamics, whereas \cite{budzinski_2022} exploits the analytic propagation of linear complex-valued dynamics between resets. 

2) Although formally a state-feedback law, $u$ is termed \emph{feedforward} to highlight its structural role in canceling out the nominal vector field without tracking error information. }
\end{remark}

\subsection{Feedforward action with SMC}
\label{subsec:ffsmc}

To overcome the requirement of unit initial magnitudes in Theorem~\ref{th:sfc} \red{and add robustness against feedforward model dependency}, we augment the controller~\eqref{eq:u_s} with an SMC term that compensates for \red{unmodeled disturbances and} any drift of the state magnitudes away from unity.

\begin{theorem} \label{th:smc_ff}
Let $u=u^1+u^2$, with $u^1$ from \eqref{eq:u_s} and $u^2$ defined component-wise by 
\begin{subequations}\label{eq:u_smc}
\begin{align}
  |u^2_k| &= \alpha \lvert\sign(|x_k|-1)\rvert, \\
  \phi_{u^2_k} &= \phi_{x_k}+\tfrac{\pi}{2}
    \left(1+\sign(|x_k|-1)\right),
\end{align}
\end{subequations} with $\alpha\in\mathbb{R}^+$. \red{Then, for any $x(0) \neq 0$, $\mathcal{M}=\{x\in\mathbb C^N: \lvert x \rvert_\bullet=\mathds{1}_N\}$ is a stable forward-invariant sliding manifold reached in a finite time bounded by 
$ T \leq \frac{\sqrt{2}}{\alpha} \lVert \lvert x(0)\rvert_\bullet -\mathds{1}_N \rVert_2$.
Moreover, for all $t\geq T$, $x(t)\in\mathcal{M}$ and the argument dynamics match \eqref{eq:real_kuramoto} under the boundary condition $\theta (T)=\phi_x(T)$.}
\end{theorem}

\begin{proof}
Under $u=u^1+u^2$, system \eqref{eq:complex_kuramoto_vectorial_u} decouples for all $k \in \{1,\ldots,N\}$ a.e. into the magnitude and phase dynamics 
\begin{subequations} \label{modarg_smc2}
\begin{align}
\label{eq:mod_smc2}  \frac{\mathrm{d}}{\mathrm{d}t}\lvert x_k \rvert = &f_k(x) +\sum_{l=1}^2\lvert u^l_k \rvert \cos\!\left(\phi_{u^l_k}-\phi_{x_k}\right), \\ 
 \label{eq:arg_smc2} \frac{\mathrm{d}}{\mathrm{d}t}\phi_{x_k} = &g_k(x) + \sum_{l=1}^2\frac{\lvert u^l_k\rvert}{\lvert x_k \rvert}\sin\!\left(\phi_{u^l_k}-\phi_{x_k}\right).
\end{align}
\end{subequations} 
\red{Substituting \eqref{eq:u_s}--\text{\eqref{eq:u_smc}} into \eqref{eq:mod_smc2} and applying Filippov regularization yields the isolated radial differential inclusion $$\frac{\mathrm{d}}{\mathrm{d}t} \lvert x_k \rvert \in -\alpha\,\operatorname{Sgn}\left(\lvert x_k \rvert-1\right).$$ Thus, $\lvert x \rvert_\bullet= \mathds{1}_N$ is a unique Filippov equilibrium for the radial subsystem, rendering $\mathcal{M}$ forward invariant for the aggregate system.
To prove finite-time stability, consider the positive-definite candidate $V\left(\lvert x \rvert_\bullet\right)=\tfrac{1}{2}\lVert \lvert x \rvert_\bullet -\mathds{1}_N \rVert_2^2$. Its set-valued time derivative along the trajectories satisfies $$\dot V \in -\alpha \left(\lvert x \rvert_\bullet-\mathds{1}_N\right)^\top\operatorname{Sgn}\left(\lvert x \rvert_\bullet -\mathds{1}_N\right).$$ For any $x \notin \mathcal{M}$, this collapses to the single scalar $\dot V = -\alpha \left\| \lvert x \rvert_\bullet -\mathds{1}_N\right\|_1$. Utilizing the norm inequality $\left\|\cdot\right\|_1~\geq~\left\|\cdot\right\|_2$, we obtain the fractional differential inequality $\dot V \leq -\alpha \left\|\lvert x \rvert_\bullet-\mathds{1}_N\right\|_2 = -\alpha \sqrt{2V}$. Direct integration verifies finite-time convergence to $\mathcal{M}$ within the bound $T$ \cite{bhat_1998}.

For the phases, substituting \eqref{eq:u_s}--\text{\eqref{eq:u_smc}} into \eqref{eq:arg_smc2} shows that the control contributions vanish for all $t \geq 0$: the $u^1$ component cancels via Theorem~\ref{th:sfc}, while the $u^2$ component evaluates to zero for $x \notin \mathcal{M}$ because $\sin\left(\tfrac{\pi}{2}(1+\operatorname{sign}(\lvert x_k \rvert-1))\right) = 0$, and identically vanishes on $\mathcal{M}$ because $|u^2_k|=0$. Consequently, for all $t \geq T$, $\lvert x \rvert_\bullet = \mathds{1}_N$ and the phase equations reduce exactly to $\dot{\phi}_{x_k} = \left. g_k(x)\right\rvert_{\lvert x \rvert_\bullet=\mathds{1}_N}$, matching the real-valued Kuramoto model \eqref{eq:real_kuramoto} under $\theta(T)=\phi_x(T)$.}
\end{proof}

\begin{remark} \label{rem_error}
1) Theorem~\ref{th:smc_ff} improves upon \cite{roberts_2008}: 
\red{finite-time convergence to the mapped real-valued Kuramoto dynamics is achieved, and is tunable via $\alpha$ without spectral gain constraints. Generally, initializing with $\theta(0)=\phi_x(0)$ does not ensure $\theta(T)=\phi_x(T)$ because of transient magnitude evolution prior to reaching the switching surface. Hence, a residual phase offset may develop between \eqref{eq:real_kuramoto} and \eqref{eq:arg_smc2}. This offset shrinks as $\alpha$ increases (reducing $T$), at the cost of control effort, and vanishes entirely when $x_0=e^{i\theta_0}$, falling back to the case of Theorem~\ref{th:sfc}.}

\red{2) Geometrically, the feedforward law $u^1$ cancels out the radial vector field to constrain trajectories to the tangent space $T_x\mathcal{M}$, yielding an exact projected dynamic. The sliding mode term $u^2$ then acts as a discontinuous exterior mechanism to guarantee finite-time attraction to this manifold.}
\end{remark}

% ================================================================
\section{Synchronization via Complex-valued SMC}
\label{sec:csmc}

The controllers of Section~\ref{sec:ff} reproduce whatever
synchronization emerges in \eqref{eq:real_kuramoto}. We now aim at
\emph{actively enforcing} synchronization at a prescribed frequency. To this end, we extend the MIMO complex-valued SMC proposed in \cite{diPalo_2023} to a class of non-autonomous systems.

Consider the complex-valued, non-autonomous system
\begin{equation}\label{eq:cs}
  \dot{x}=F(x,t)+G(x,t)u,\quad x\in\mathbb{C}^N,\;u\in\mathbb{C}^M,
\end{equation}
where $F\!:\!\Omega \times \mathbb{R} \to \mathbb{C}^N$ and $G\!:\!\Omega \times \mathbb{R} \to \mathbb{C}^{N \times M}$ with $\Omega \subseteq \mathbb{C}^N$, are assumed to be holomorphic in $x$ and continuous in $t$. Let the complex switching manifold be defined by
\begin{equation} \label{sman}
    \mathcal{S}=\{(x,t) \in \Omega \times \mathbb{R}_{\geq 0}: s(x,t)=0\},
\end{equation}
where the complex switching function $s\!:\!\Omega \times \mathbb{R} \to \mathbb{C}^M$ is also assumed holomorphic in $x$ and continuous in $t$. Finally, we define the matrix-valued function $D(x,t)=\tfrac{\partial s}{\partial x}G(x,t)$.

\begin{property} \label{ass:switchin_1}
\red{The matrix} $D(x,t)$ is diagonal and there exist $\Omega_1 \subseteq \Omega$ with $\Omega_1 \times \mathbb{R}_{\geq 0} \cap \mathcal{S} \neq \emptyset$, a gain vector $\red{\mathcal{K}} \in \mathbb{C}^M$, and constants $\epsilon_1,\epsilon_2 \in \mathbb{R}^+$ such that, for all $(x,t) \in \Omega_1 \times\mathbb{R}_{\geq 0}$,
\begin{subequations} \label{eq:cond}
    \begin{align} 
        \label{eq:cond0}  & \lvert D_{\red{kk}}(x,t) \rvert \geq \epsilon_1, \\
        \label{eq:cond1}  & \lvert \red{\mathcal{K}_k} \rvert \cos\!\left(\phi_{D_{\red{kk}}}(x,t)+\phi_{\red{\mathcal{K}_k}}\right)-\lvert u_{\mathrm{eq}_\red{k}}(x,t) \rvert \geq \epsilon_2,
    \end{align}
\end{subequations} 
for all $\red{k}=1,\ldots,M$, where $u_{\mathrm{eq}} \in \mathbb{C}^M$ denotes the equivalent control \cite{edwards_1998} defined by
\begin{equation} \label{eq:ueq}
    u_{\mathrm{eq}}(x,t) =-D^{-1}(x,t) \left(\tfrac{\partial s}{\partial x}(x,t)F(x,t) + \tfrac{\partial s}{\partial t}(x,t)\right).
\end{equation}
\end{property}

\begin{theorem}\label{th:general_smc}
Let \red{Property} \ref{ass:switchin_1} hold. Then, the switched controller 
\begin{equation} \label{eq:cu}
    u =-\diag\left(\red{\mathcal{K}}\right) \sign(s)
\end{equation}
induces a sliding motion of \eqref{eq:cs} on $\Omega_1 \times \mathbb{R}_{\geq 0} \cap \mathcal{S}$. Moreover, \red{for all  $\left(x(0),0\right)$ sufficiently close to $\Omega_1 \times \mathbb{R}_{\geq 0} \cap\mathcal{S}$,}  $\mathcal{S}$ is reached in a finite time $T \leq \frac{\sqrt{2}}{\epsilon_1\epsilon_2} \lVert s\left(x(0),0\right) \rVert_2$.
\end{theorem}

\begin{proof}
The proof follows the argument of Theorem~4 in \cite{diPalo_2023}, \red{interpreting the closed-loop trajectories in the sense of Filippov solutions to rigorously validate the sliding dynamics}. The extension to the
non-autonomous setting requires two modifications relative to
the original autonomous case. First, the time derivative of the
switching function $s(x,t)$ now contains the explicit partial
derivative $\partial s/\partial t$, which contributes to the
equivalent control \eqref{eq:ueq} and must be dominated by the gain $\red{\mathcal{K}}$. Second, the reaching-time bound must
absorb this additional term, which is accommodated through
condition~\eqref{eq:cond1} via the constant $\epsilon_2$. \red{Notice that \eqref{eq:cond0} in Property \ref{ass:switchin_1} is necessary for the existence of sliding motion, and \eqref{eq:cond1} is sufficient to have a nonempty sliding domain \cite{edwards_1998}.} 
\end{proof}

Building on Theorem \ref{th:general_smc}, we can now introduce a control strategy to actively enforce synchronization at a prescribed frequency $\bar{\omega}$. To this end, we choose 
\begin{equation} \label{eq:switching_functions}
    s(x,t) \coloneqq x-e^{i \bar{\omega} t} \mathds{1}_N,
\end{equation}
which is holomorphic in $x$ and continuous in $t$. On the associated switching manifold $\mathcal{S}$ defined using \eqref{sman}, it holds for all $k\in\{1,\ldots,N\}$ that $\lvert x_k \rvert=1$ and $\phi_{x_k}=\bar{\omega} t$, indeed having all the oscillators  phase locked at frequency $\bar{\omega}$. \red{The following theorem provides a sufficient gain bound to fulfill Property \ref{ass:switchin_1} and guarantee finite-time synchronization.}

\begin{theorem} \label{th:smc_particularized}
\red{Let us assume that the graph $\mathcal{G}$ is simple. Moreover, let $\red{\mathcal{K}} \in \mathbb{R}^N$ satisfy
\begin{equation} \label{eq:Ki} 
    \red{\mathcal{K}_k} \red{>} \lvert\omega_k\rvert+\bar{\omega}+\sigma \left(N-1\right), \; \red{k}=1,\ldots,N.
\end{equation}
Then, for all  $\left(x(0),0\right)$ sufficiently close to $S$,} the discontinuous control law \eqref{eq:cu} with switching surface defined by \eqref{eq:switching_functions} induces phase locking at frequency \red{$\bar{\omega} \in  \mathbb{R}_{>0}$} for the complex-valued Kuramoto system \eqref{eq:complex_kuramoto_vectorial_u} in finite time
$T_s \leq \frac{\sqrt{2}}{\epsilon_2} \lVert x(0)-\mathds{1}_N \rVert_2$.
\end{theorem}

\begin{proof}
We map \eqref{eq:complex_kuramoto_vectorial_u} to the general system \eqref{eq:cs} by defining $F(x,t)=\left(i \diag(\omega)+\sigma A\right)x$, $G(x,t)=\mathbb{I}_N$. Both vector fields are holomorphic in $x$ and, being autonomous, are trivially continuous in $t$. Since $\frac{\partial s}{\partial x}(x,t)=\mathbb{I}_N$, it follows that $D(x,t) = \mathbb{I}_N$. Thus, condition \eqref{eq:cond0} is satisfied with $\epsilon_1=1$. 

Let $a_k^\top$ denote the $k$-th row of the adjacency matrix $A$. Since the underlying graph is assumed to be simple, $a_{kk}=0$. The equivalent control \eqref{eq:ueq} reduces to $u_{\mathrm{eq}}(x,t) = i \bar{\omega}e^{i \bar{\omega}t} \mathds{1}_N - \left(i \diag\left(\omega\right)+\sigma A\right)x$. Hence, $\forall\,k=1,\ldots,N$, and $(x,t) \in \mathcal{S}$:
\begin{align*}
& \left| u_{\mathrm{eq}_k}(x,t) \right|  \leq \bar{\omega} + \lvert e^{i \bar{\omega}t}\rvert + \lvert \omega_k\rvert \lvert x_k\rvert+\sigma\red{\lvert a_k^\top x\rvert} \leq \\ & \leq \bar{\omega}  + \lvert \omega_k\rvert + \red{\lvert|a_k\rvert|_1 \lvert|\lvert x \rvert_\bullet\rvert|_\infty} \leq \bar{\omega}  + \lvert \omega_k\rvert + \sigma \left(N-1\right),
\end{align*}
\red{because $\lvert|a_k\rvert|_1\leq N-1$, and $\lvert|\lvert x \rvert_\bullet\rvert|_\infty =1$ on $\mathcal{S}$}. So \eqref{eq:cond1} holds with $\red{\epsilon_2 := \min_k\left\{\mathcal{K}_k-\lvert\omega_k\rvert- \bar{\omega} -\sigma \left(N-1\right)\right\}}$. Therefore, \red{Property}~\ref{ass:switchin_1} is satisfied and by Theorem \ref{th:general_smc}, a sliding motion is induced on $\mathcal{S}$, which is reached in finite time $T_s$ \red{satisfying the simplified bound since $\epsilon_1=1$}. 
\end{proof}

\begin{remark}
The gain condition \eqref{eq:Ki} is sufficient but conservative,
derived assuming all-to-all coupling. Phase locking can be
achieved with smaller gains on sparse networks.
\end{remark}

% ================================================================
\section{Numerical Validation}
\label{sec:num}

\begin{figure}[t]
  \centering
  \begin{tabular}{cc}
    \includegraphics[width=0.47\columnwidth]{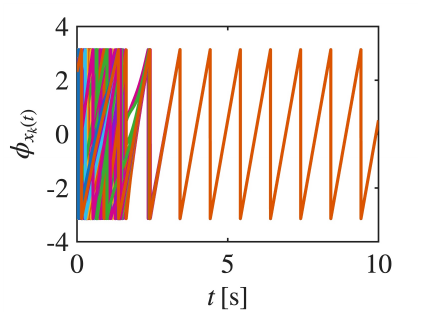} & 
    \includegraphics[width=0.47\columnwidth]{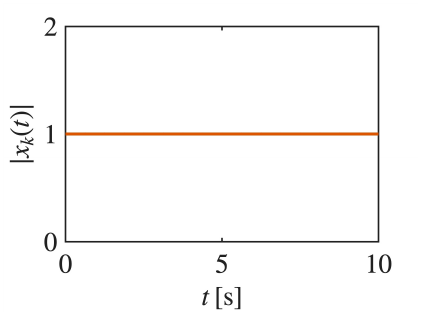} \\ [-1mm]
    \small(a) & \small(b)\\[-1mm]
    \includegraphics[width=0.47\columnwidth]{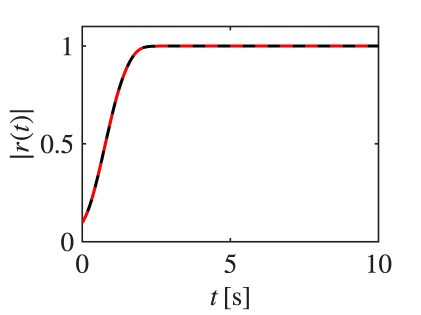} &
    \includegraphics[width=0.47\columnwidth]{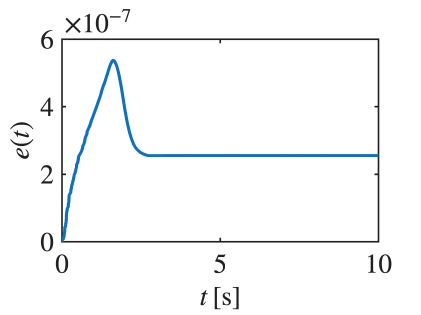} \\ [-1mm]
    \small(c) & \small(d)\\[-1mm]
  \end{tabular}
  \caption{Switched feedforward (Theorem~\ref{th:sfc}).
    (a)~Complex-valued arguments;
    (b)~magnitudes $\equiv1$; 
    (c)~order parameter for real (black) and  complex-valued  (red) Kuramoto models;
    (d)~\red{mean circular phase error} $\approx0$.}
  \label{fig:ff}
\end{figure}

The proposed control strategies are validated on an Erd\"{o}s-R\'{e}nyi network of $N=100$ oscillators with edge probability $p=0.2$, identical natural frequencies $\omega_k=2\pi\,\mathrm{rad}/\mathrm{s}$ for all $k$, and coupling strength $\sigma=0.25$. \red{This yields synchronization of the real-valued Kuramoto model.} Phase correspondence between the complex-valued model and the real-valued Kuramoto dynamics is evaluated using the \red{mean circular phase error $e(t) = \frac{1}{N}\sum_{k=1}^N \left| \operatorname{wrap}_{[-\pi,\pi)} \left(\phi_{x_k}(t)-\theta_k(t)\right) \right|$, where \(\operatorname{wrap}_{[-\pi,\pi)}(\cdot)\) maps phase differences to principal values.}

\subsection{Switched feedforward control} \label{subsec:validation_ff}
For the switched feedforward controller of
Theorem~\ref{th:sfc}, the real-valued model is initialized
with phases $\theta_0$ drawn uniformly from $(-\pi,\pi)$,
while the complex-valued model is initialized at
$x_0=e^{i\theta_0}$, ensuring unit initial magnitudes.

The results are shown in Figure~\ref{fig:ff}. The complex-valued model achieves complete phase locking (panel~a) and reproduces the behavior of the real-valued model, as confirmed by the synchronization order parameters (panel~c). Consistent with Theorem~\ref{th:sfc}, the state magnitudes remain equal to one (panel~b), while the \red{mean circular phase error} vanishes to numerical precision for all $t\geq 0$ (panel~d). Hence, the two models exhibit exact continuous-time phase correspondence, improving on the hybrid reset approach of \cite{budzinski_2022}, where matching occurs only at discrete reset instants.

\subsection{Feedforward action with SMC} \label{subsec:validation_ff_smc}
The FF\,+\,SMC controller of Theorem~\ref{th:smc_ff} is
validated with gain $\alpha=10$. The initial phases are randomly chosen as in Section \ref{subsec:validation_ff} while the initial moduli of the
complex-valued model are drawn from the uniform distribution $\mathcal{U}(0,2)$ rather than set to unity, to test the SMC component.

Figure~\ref{fig:ffsmc} shows that the phases of the complex-valued model achieve phase locking (panel a), while the magnitudes converge to the sliding surface $\{\lvert x \rvert_\bullet=\mathds{1}_N\}$ in finite time, reaching unity well within the theoretical upper bound $T=0.78$\,s given by Theorem \ref{th:smc_ff} (panel~b). \red{Consistently with Remark \ref{rem_error}, a residual constant phase offset with respect to the corresponding real-valued model arises (panel~d).} Compared with the state-feedback approach of \cite{roberts_2008}, this scheme achieves finite-time rather than asymptotic \red{mapping} and requires no spectral tuning of the feedback gains.

\begin{figure}[t]
  \centering
  \begin{tabular}{cc}
    \includegraphics[width=0.47\columnwidth]{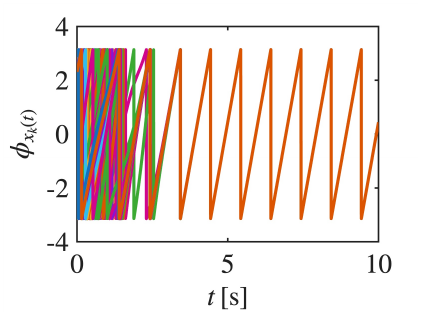} &
    \includegraphics[width=0.47\columnwidth]{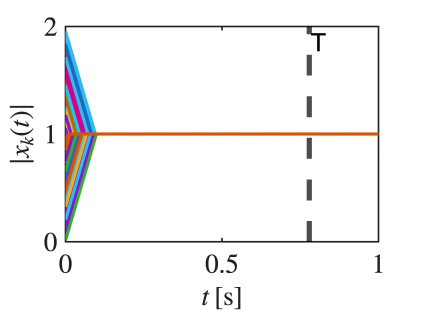} \\ [-1mm]
    \small(a) & \small(b) \\ [-1mm]
    \includegraphics[width=0.47\columnwidth]{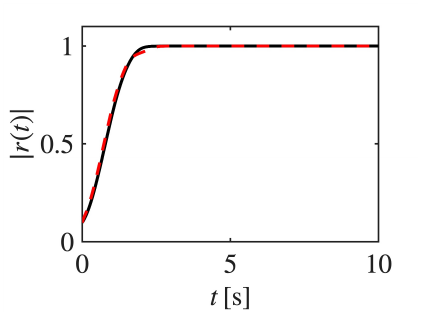} &
    \includegraphics[width=0.47\columnwidth]{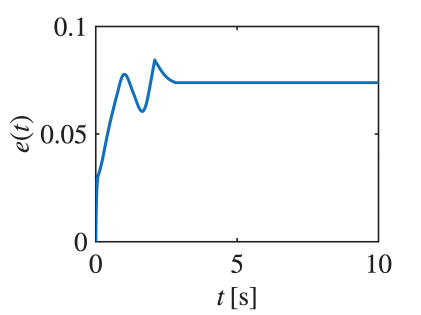}\\[-1mm]
    \small(c) & \small(d)\\[-1mm]
  \end{tabular}
  \caption{FF\,+\,SMC (Theorem~\ref{th:smc_ff}). 
    (a)~Complex-valued arguments;
    (b)~magnitudes converging to $1$;
    (c)~order parameter for real (black) and  complex-valued (red) Kuramoto models;
    (d)~\red{mean circular phase error} showing phase offset.}
  \label{fig:ffsmc}
\end{figure}

\subsection{Synchronization via complex-valued SMC}

The complex SMC of Theorem~\ref{th:smc_particularized} is evaluated for a desired synchronization frequency $\bar{\omega}=4\pi\,\mathrm{rad}/\mathrm{s}$,
with control gains $\red{\mathcal{K}_k}=50$ for all $k$, thereby satisfying condition~\eqref{eq:Ki}. Initial conditions are chosen as in
Section~\ref{subsec:validation_ff_smc}.

Figure~\ref{fig:smc_particularized} summarizes the results. Panel~(a) shows that the arguments of the controlled complex-valued network
achieve phase locking at the prescribed frequency $\bar{\omega}$, while panel~(b) confirms that all magnitudes converge to unity in finite time, with the sliding manifold reached within the theoretical bound $T_s=1.22$\,s \red{given in Theorem \ref{th:smc_particularized}}. The
synchronization order parameter in panel~(c) further confirms fast finite-time convergence; for comparison, it also includes the order parameter of the corresponding real-valued Kuramoto model, which synchronizes under the same conditions but at its natural frequency $2\pi\,\mathrm{rad}/\mathrm{s}$ rather than the prescribed $4\pi\,\mathrm{rad}/\mathrm{s}$.

The robustness of the proposed strategy is assessed in panel~(d), where the natural frequencies are drawn from a normal distribution $\mathcal{N}(2\pi,1)$ and the coupling strength is reduced to $\sigma=0.1$. Under these heterogeneous conditions, the real-valued Kuramoto model fails to synchronize, yet the SMC successfully enforces phase locking across the entire network at the prescribed frequency. This result highlights the fundamental advantage of the complex-valued control formulation: by operating on the linear embedding rather than the nonlinear phase equations, the controller can impose synchronization even in regimes where the uncontrolled real-valued model does not synchronize spontaneously.

\begin{figure}[t]
  \centering
  \begin{tabular}{cc}
    \includegraphics[width=0.47\columnwidth]{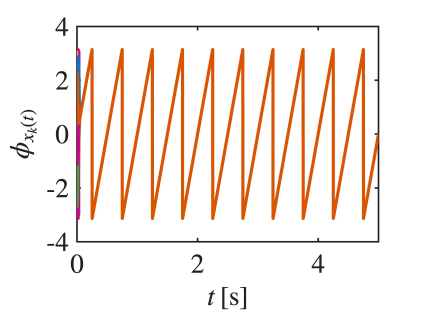} &
    \includegraphics[width=0.47\columnwidth]{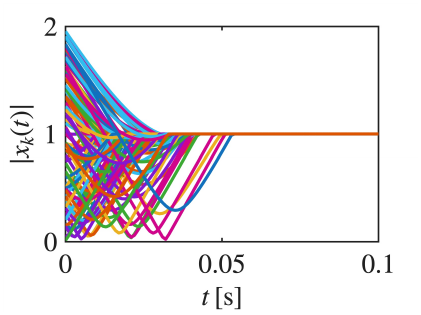}\\[-1mm]
    \small(a) & \small(b)\\[-1mm]
    \includegraphics[width=0.47\columnwidth]{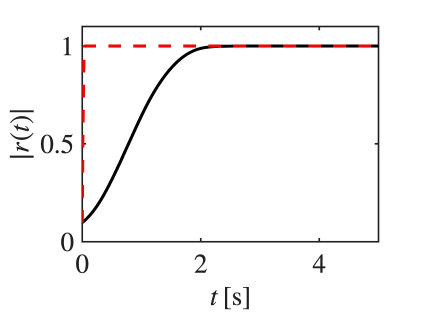} &
    \includegraphics[width=0.47\columnwidth]{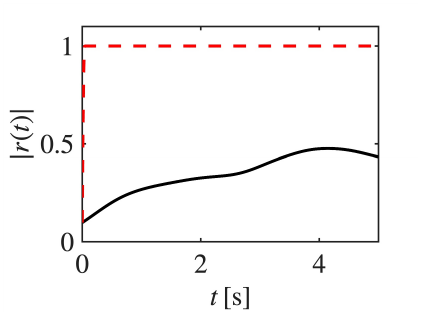}\\[-1mm]
    \small(c) & \small(d) \\ [-1mm]
  \end{tabular}
  \caption{Complex SMC (Theorem~\ref{th:smc_particularized}).
    (a)~Arguments locked at $\bar{\omega}=4\pi$;
    (b)~magnitudes converging to $1$;
    (c)~order parameter for real (black) and  complex-valued (red) homogeneous Kuramoto models;
    (d)~order parameter for real (black) and  complex-valued (red) heterogeneous Kuramoto models.}
  \label{fig:smc_particularized}
\end{figure}

% ================================================================
\section{Conclusions}
\label{sec:conc}
This work \red{established a shared control perspective for complex-valued Kuramoto oscillator networks which act as higher-order state space embeddings designed to map onto the real-valued Kuramoto dynamics. Then, we} proposed three switched control strategies that overcome key limitations of existing methods. By interpreting the state-feedback approach of \cite{roberts_2008} and the hybrid reset method of \cite{budzinski_2022} as mechanisms to enforce a common state modulus, we identified an underlying control objective and used it to guide improved designs.

Within this framework, the switched feedforward law (Theorem~\ref{th:sfc}) achieves exact continuous-time phase correspondence with classical Kuramoto dynamics for all $t\ge0$, eliminating reset-based drift mismatch. The combined feedforward-SMC law (Theorem~\ref{th:smc_ff}) relaxes unit-magnitude initialization and ensures finite-time \red{mapping} with an explicit, tunable bound, improving upon the asymptotic result of \cite{roberts_2008} without spectral tuning. Finally, the prescribed-frequency SMC law (Theorem~\ref{th:smc_particularized}), based on a novel non-autonomous extension of complex-valued MIMO sliding-mode control \red{and a tailored control gain bound}, enforces phase locking at a desired frequency in finite time, independent of natural frequencies and coupling strengths. This enables synchronization of heterogeneous networks where the classical Kuramoto model fails.

Future work will explore output-feedback designs, pinning control to reduce reference-awareness, generalized quotient-space coupling, and directed or time-varying networks. We also target tighter gain conditions via sparse graph spectra, explicit domain of attraction estimates for Theorems \ref{th:general_smc} and \ref{th:smc_particularized}, and power systems or neuromorphic computing applications.

% ================================================================
\section*{Acknowledgments}
The authors thank Prof.~Marco Coraggio and the anonymous reviewers for their valuable comments and suggestions.

% ================================================================
\bibliographystyle{IEEEtran}
\bibliography{references}

\end{document}